%% file: iclr2025_conference.tex
\documentclass[conference]{ieeeconf} 
\IEEEoverridecommandlockouts

\usepackage{etoolbox}
\newbool{techreport}
\booltrue{techreport}

\usepackage{silence}
\ActivateWarningFilters[captions]
\usepackage{subcaption}
\DeactivateWarningFilters[captions]

\usepackage{amsmath, amssymb, amsthm, amsfonts} 
\usepackage{mathrsfs}                           
\usepackage{bm}                                 
\usepackage{graphicx}                           
\usepackage{booktabs}                           
\usepackage{multirow}                           
\usepackage{hyperref}                           
\usepackage{url}                                
\usepackage{xcolor}                             
\usepackage{algorithm}                          
\usepackage{algorithmicx}                       
\usepackage{algpseudocode}                      
\usepackage{float}                              
\usepackage{comment}                            
\usepackage{subcaption}
\usepackage{caption}
\usepackage{mathtools}
\mathtoolsset{showonlyrefs}
\usepackage{algorithm}
\usepackage{algpseudocode} 

\usepackage{cite}

\usepackage{caption}
\usepackage{float} 

\usepackage{float}    
\usepackage{placeins} 

\usepackage{tikz}
\usepackage{tikz-3dplot}
\usepackage{pgfplots}
\pgfplotsset{compat=1.18}
\usepgfplotslibrary{groupplots}
\usepackage{graphicx}

\newlength{\figwidth}
\newlength{\figheight}
\definecolor{pastelBlue}{RGB}{163,199,232}
\definecolor{pastelSand}{RGB}{226,211,165}

\newtheorem{theorem}{Theorem}
\newtheorem{proposition}{Proposition}
\newtheorem{lemma}{Lemma}
\newtheorem{definition}{Definition}
\newtheorem{assumption}{Assumption}[]
\newtheorem{problem}{Problem}
\newtheorem{remark}{Remark}

\newcommand{\NE}{\mathrm{NE}}
\newcommand{\WNE}{\mathrm{WNE}}

\newcommand{\cone}{\mathrm{cone}}

\DeclareMathOperator*{\argmin}{arg\,min}

\input{math_commands.tex}
\ifbool{techreport}{
\title{Online Scalarization in Vector-Valued Games}
}
{
\title{Online Scalarization in Vector-Valued Games}
}

\author{Ehsan Asadollahi, Calvin Hawkins, Matthew Hale\textsuperscript{1}\thanks{\textsuperscript{1}School of Electrical and Computer Engineering, Georgia Institute of Technology, Atlanta, GA, USA.
Emails: {\texttt{\{sasadoll3,chawkins64,mhale30\}@gatech.edu}}}
\thanks{
This work was supported by ONR under grant N00014-24-1-2432.}}

\begin{document}

\maketitle

\input{sec/0_abstract}    
\input{sec/1_Introduction}
\input{sec/2_background}

\input{sec/3_Problem_Setup}
\input{sec/4_algorithms}
\input{sec/5_regretbounds}
\input{sec/6_Empirical_Validation}

\input{sec/7_Conclusion}

\bibliography{iclr2025_conference}
\bibliographystyle{ieeetr}

\ifbool{techreport}{
  \input{app/app_A}
}{
}



\end{document}

%% file: sec/0_abstract.tex
\begin{abstract}
We study repeated multi-player vector-valued games in which a player observes 
a payoff vector each round and evaluates outcomes through linear scalarizations
of those vectors. Different from most prior works, the choice of scalarization is treated as an online decision variable rather than a fixed modeling decision. We propose a bi-level learning framework in which an outer learner chooses a scalarization from a finite candidate class on a slow timescale, while a faster inner bandit no-regret learner selects actions using the scalar feedback induced by the chosen scalarization. Performance of this approach is defined with respect to a certain true weight vector, and the deployed scalarizations act as control signals that shape the induced payoff trajectory. We provide implementable algorithms based on bandit online mirror descent with stabilized importance weighting, and we derive finite-time performance guarantees in the form of sublinear regret bounds. Experiments on a vector-valued extension of a canonical game show that convergence to the preferred equilibrium rises from roughly $50\%$ under non-adaptive scalarization to about $80\%$ under our proposed method. 

\end{abstract}

%% file: sec/1_Introduction.tex
\section{Introduction}

Multi-agent systems often have more than one objective, and agents may operate under multiple criteria such as safety, efficiency, fairness, robustness, and compliance, which may not be naturally commensurate~\cite{RVWD13,Hay21,Kha24}. Such formulations naturally lead to vector-valued measures of performance, where each decision produces a multi-dimensional outcome and agents must evaluate that outcome according to their own preferences. In many settings, 
vector-valued decision-making proceeds by scalarizing vector-valued outcomes using a fixed weight vector that specifies the relative importance of the different objectives.

We study such vector-valued outcomes in the setting of repeated games, but
rather than treating a scalarization as a fixed modeling decision, we treat it as an algorithmic control knob that a player updates across repeated interactions. The trajectory of play depends on the choice of weight vector because that choice affects the player's actions and thus how opponents respond to those actions. 
As a result, the time-varying weight vector used in a player's learning rule can be chosen to shape its interactions.


We study this approach in a multi-player repeated game with vector-valued outcomes, and we design
strategies for a choice of focal player. 
The focal player evaluates performance using a fixed objective weight vector that defines its performance criterion, 
but it may deploy alternative weight vectors over time.
The reason for using different weight vectors is that we do not
use a fully adversarial minimax model in which opponents can select responses to negate any potential benefit from such  reweighting. In that worst-case setting, adapting weights on objectives cannot be expected to improve worst-case performance. 

Our aim is instead to understand how adaptation can improve performance in responsive environments in which the opponent's behavior depends on the observed outcomes. For example, the opponent 
may itself be a learning agent, such as a no-regret learner. This formulation creates an endogenous coupling: changing the focal player's deployed weight vector changes the player's actions, which can change the opponents' responses, which changes future outcomes and ultimately changes the player's performance with respect to its objective. Consequently, simply deploying the objective weight vector need not be optimal in finite time. 

We therefore search over weight vectors 
using a bi-level learning architecture. The inner layer is an action-selection learner that operates under the deployed weight vector and updates a mixed action distribution from bandit feedback, meaning that it observes only the payoff of the action actually played rather than the payoffs of all actions. The outer layer selects the deployed weight vector and updates a distribution over deployed weight vectors on a slower timescale. The outer layer is evaluated using outcomes aggregated over time and scored under the fixed objective weight vector, while the inner layer uses per-round bandit feedback computed under the deployed weight vector. This separation between shaping and evaluation enables deployed-weight selection to influence the coupled learning dynamics without changing the focal player's ultimate objective.

To summarize, our contributions are as follows: 
\begin{itemize}
\item We propose a bi-level learning architecture for online scalarization in vector-valued games 
(Algorithms~\ref{alg:outer_weight} and~\ref{alg:inner_action}). 

\item 
We derive finite-time regret bounds for both levels of the algorithm, 
and we show sublinear regret for both (Theorems~\ref{thm:inner_block_regret} and~\ref{thm:outer_block_regret}). 

\item We demonstrate experimentally in a representative two-equilibrium game that convergence to the objective-favored equilibrium increases from roughly \(50\%\) under standard no-regret learning to about \(80\%\) under our method (Section~\ref{sec:experiments}).
\end{itemize}

\subsection{Related Work}
One classical viewpoint 
on vector-valued games 
is set-based and uses Blackwell's approachability framework. In repeated vector-valued interactions in this framework, a player seeks to drive the time-average payoff toward a target set $C \subseteq \mathbb{R}^d$~\cite{Blk56}. A set $C$ is ``approachable'' if there exists a strategy such that the time-average payoff converges to $C$ against any opponent~\cite{Per18,ABH11,Shi16,PM13,GM25,FKS21,MT06}.
In contrast, we do not study robust convergence to a fixed target set. Instead, we study online selection of scalarizations as a mechanism for shaping equilibrium outcomes.


Another classical viewpoint is direction-based and obtains a scalar decision criterion from vector outcomes through linear scalarization~\cite{RVWD13,Ehr05,Mie99,boyd04}: a player selects a weight vector and evaluates outcomes by projection onto that direction. This approach produces a family of induced scalar games indexed by the choice of weight vectors. 
Under suitable regularity conditions, changing the scalarization can recover different weak equilibria of the original vector-valued game~\cite{CKR25}.
This equilibrium correspondence highlights a key modeling point: changes
in weight vectors can induce changes in the effective incentives and thereby 
induce changes in players' best responses and the resulting equilibria. We act on this idea by developing a bi-level learning algorithm in which the player adaptively selects its weight vector online in order to shape the induced scalar game and, in turn, the resulting equilibrium outcomes.

The rest of the paper is organized as follows. Section~\ref{sec:background} reviews vector-valued games and states the problems we solve.
Section~\ref{sec:algorithms} presents the bi-level learning protocol. Section~\ref{sec:theory} develops finite-time regret guarantees. Section~\ref{sec:experiments} reports empirical results, and Section~\ref{sec:conclusion} concludes.

\textbf{Notation:}
We use~$\mathbb{R}$ to denote the reals and~$\mathbb{N}$ to denote the positive integers. 
For $n\in\mathbb{N}$, 
we define
$[n]:=\{1,\dots,n\}$. Given a collection $(z_i)_{i\in N}$, we use $z_{-i}:=(z_j)_{j\neq i}$ to denote the 
sub-collection of all components except the $i^{th}$, and we use $(z_i;z_{-i})$ to denote the 
collection formed by combining $z_i$ with $z_{-i}$. We use $\|x\|_p$ to denote the $\ell_p$-norm
of a vector~$x$, and $B_p(r):=\{x:\|x\|_p\le r\}$ is the ~$\ell_p$-norm closed ball of radius~$r$
centered on the origin. 
The notation $\mathrm{conv}(\mathcal{S})$ denotes the convex hull of a set~$\mathcal{S}$. 
For a finite set $\mathcal{S}$, the notation $\Delta(\mathcal{S})$ denotes the probability simplex over $\mathcal{S}$. We denote by \(\cone(A)\) the conic hull of a set \(A\), i.e., the set of all finite non-negatively weighted linear combinations of elements of \(A\). We use \(S^{d-1}:=\{x\in\mathbb{R}^{d}:\|x\|_2=1\}\) 
to denote the unit sphere in \(\mathbb{R}^{d}\).

%% file: sec/2_background.tex
\section{Background and Problem Statements}
\label{sec:background}

This section provides background on vector-valued games and then gives problem statements. 

\subsection{Vector-Valued Games}
\label{subsec:vector_valued_games}

We study repeated interactions in a vector-valued game among
$n \in \mathbb{N}$ players. Let
\begin{equation} \label{eq:Ndef}
N=[n]
\end{equation}
be the set of their indices, and consider $T\in\mathbb{N}$ rounds of the game. 
At each round \(t\in[T]\), each player \(i \in N\) chooses an action
\begin{equation} \label{eq:Aidef}
a_{i,t}\in\mathcal{A}_i,
\end{equation}
where each~$\mathcal{A}_i$ is a finite non-empty set. 
The payoff vector to player \(i\) 
for these actions is
\begin{equation}
\label{eq:stage_payoff_i}
u_{i,t} := u_i(a_{i,t},a_{-i,t})\in\mathbb{R}^{d_i},
\quad
u_i:\mathcal{A}_i\times\mathcal{A}_{-i}\to\mathbb{R}^{d_i}
\end{equation}
for some~$d_i \in \mathbb{N}$. 
We write the time-averaged payoff over~$t~\in~[T]$ as
\begin{equation}
\label{eq:avg_vector_payoff_i}
\bar u_{i,T} := \frac{1}{T}\sum_{t=1}^T u_{i}(a_{i,t},a_{-i,t}).
\end{equation}

\subsubsection{Cone-induced preferences and equilibria}
For each~$i~\in~N$, player~$i$ 
 chooses a strategy \(x_i \in X_i\), where \(X_i\) is a nonempty compact convex set. 
 The joint strategy profile is 
 \begin{equation} \label{eq:Xdef}
 x=~(x_1,\dots,x_n) \in X:=\prod_{i\in N} X_i. 
 \end{equation} 
If the joint strategy~$x$ is used at time~$t$, then player~$i$ receives
 the payoff~$u_{i, t}(x)$.

In a repeated game, the action set \(\mathcal{A}_i\) determines the set of actions available to player \(i\) at each round, and the strategy set \(X_i\) for player \(i\) contains the set of decision rules available to that player. Thus, a strategy
\begin{equation} \label{eq:Xidef}
x_i \in X_i
\end{equation}
may be interpreted as a rule for selecting actions from \(\mathcal{A}_i\). 
For example, if player \(i\) uses mixed actions, then \(X_i\) may be the simplex over the action set \(\mathcal{A}_i\). 

Preferences over payoff vectors are modeled by a cone-induced partial order. A non-empty 
set \(K\subseteq\mathbb{R}^{d}\) is called a convex cone if \(\alpha x + \beta y \in K\) for all \(x,y\in K\) and all \(\alpha,\beta\ge 0\). For each~$i \in N$, 
let 
\begin{equation} \label{eq:Kidef1}
K_i \subseteq \mathbb{R}^{d_i}
\end{equation}
be a non-empty, closed, convex cone that encodes player \(i\)'s preferences. 
Then, for \(a,b\in\mathbb{R}^{d_i}\) we write
\begin{equation} \label{eq:Kidef}
a \le_{K_i} b
\quad \textnormal{ if and only if } \quad
b-a \in K_i,
\end{equation}
which means that \(b\) is weakly preferred to \(a\) 
(namely, $b$ is no worse than~$a$)
under the order induced by \(K_i\). We write \(a <_{K_i} b\) if \(b-a \in \operatorname{int}(K_i)\),
and then~$b$ is preferred to~$a$ because it is strictly better than~$a$. 

A vector-valued game takes the form 
\begin{equation}
\label{eq:Vector_valued_Game}
G := \bigl(N,\{\mathcal{A}_i\}_{i\in N}, \{u_i\}_{i\in N}, \{X_i\}_{i\in N}, \{K_i\}_{i\in N}\bigr),
\end{equation}
where~$N$ is from~\eqref{eq:Ndef},
$\mathcal{A}_i$ is from~\eqref{eq:Aidef},
$u_i$ is from~\eqref{eq:stage_payoff_i},
$X_i$ is from~\eqref{eq:Xidef},
and $K_i$ is from~\eqref{eq:Kidef1}.

We now define equilibria for vector-valued games. 

\begin{definition}[Equilibrium notions in a vector-valued game {\cite{CKR25}}]
\label{def:vector_game_equilibria}
Consider~$X$ from~\eqref{eq:Xdef}. An element~$x^\star~\in~X$ is called 
\begin{enumerate}
\item[(i)] a Nash equilibrium if for all \(i \in N\) and all \(x_i\in X_i\),
the condition $u_i(x_i;x_{-i}^\ast) \le_{K_i} u_i(x_i^\ast;x_{-i}^\ast)$
implies that
$u_i(x_i^\ast;x_{-i}^\ast) \le_{K_i} u_i(x_i;x_{-i}^\ast)$.
\item[(ii)] a weak Nash equilibrium if for all \(i \in N\) there does not exist \(x_i\in X_i\) such that
$u_i(x_i;x_{-i}^\ast) <_{K_i} u_i(x_i^\ast;x_{-i}^\ast)$.
\end{enumerate}
We denote the corresponding sets 
of equilibria
by \(\NE(G)\) and \(\WNE(G)\), 
respectively, 
and we have~\(\NE(G)~\subseteq~\WNE(G)\).
\end{definition}

In words, a Nash equilibrium is a profile at which each player’s outcome is maximal with respect to the cone order among all unilateral deviations, whereas a weak Nash equilibrium is a profile at which no player can obtain a strictly cone-better outcome by deviating unilaterally.

\subsubsection{Linear scalarization and equilibrium correspondence}
\label{subsubsec:linear_scalarization}

A linear scalarization converts vector payoffs into scalar utilities by projecting vector payoffs onto a weight vector. 
For each~\(i \in N\), consider player~\(i\)'s closed convex preference cone~\(K_i\) and its induced order~\(\le_{K_i}\) from~\eqref{eq:Kidef}. For each~$i \in N$, 
the admissible weights are naturally tied to 
player~$i$'s cone-induced order through the dual cone of~$K_i$. The Euclidean dual cone of \(K_i\) is
\[
K_i^* := \{\psi_i \in \mathbb{R}^{d_i} : \langle \psi_i, y\rangle \ge 0 \ \text{for all } y\in K_i\}.
\]
If a joint strategy profile \(x \in X\)  is used, then player \(i\) receives the vector payoff \(u_i(x)\). For each \(\psi_i \in K_i^* \setminus \{0\}\), the corresponding linear scalarization is
\[
f_i^{\psi_i}(x) := \langle \psi_i, u_i(x)\rangle.
\]
Thus, if a strategy profile \(x\) is played, then player \(i\)'s vector payoff \(u_i(x)\) is evaluated through the weight vector \(\psi_i\).

Now fix a weight profile \(\psi = (\psi_1,\dots,\psi_n)\) with \(\psi_i~\in~K_i^*~\setminus~\{0\}\) for all \(i \in N\). This profile induces a scalarized game in which player \(i\)'s payoff is \(f_i^{\psi_i}(x)\). We write \(\NE(\psi)\) for the set of Nash equilibria of this scalarized game.
Restricting weights \(\psi_i\) to \(K_i^* \setminus \{0\}\) ensures that the induced scalarization is monotone with respect to the preference order induced by \(K_i\), which we prove in Section~\ref{sec:algorithms}.


\begin{proposition}[Scalarization characterization of weak Nash equilibria {\cite{CKR25}}]
\label{prop:wne_union_scalar_ne}
For the vector-valued game \(G\), the set of weak Nash equilibria satisfies
\[
\WNE(G)
=
\bigcup_{\psi \in \prod_{i=1}^n (K_i^*\setminus\{0\})}
\NE(\psi).
\]
\end{proposition}

Proposition~\ref{prop:wne_union_scalar_ne} implies that weight vectors can be used as an equilibrium-selection device: changing \(\psi\) changes the induced scalar 
game and hence changes which equilibria in \(\NE(\psi)\) may be reached, 
while \(\WNE(G)\) is determined by~\(G\).


\subsection{Problem Formulation} \label{subsec:problem_setup}
Proposition~\ref{prop:wne_union_scalar_ne} shows that changing the scalarization can change which equilibria are reached. Motivated by this equilibrium correspondence, we distinguish between two roles for weight vectors in our setting: evaluation and deployment.

For a fixed player \(i \in N\), let
\begin{equation}
\label{eq:obj_psi}
\psi_i^{\text{obj}} \in K_i^* \setminus \{0\}
\end{equation}
denote the fixed objective weight vector that defines the performance criterion used to evaluate player~$i$'s realized average vector payoff. Over a horizon \(T\), player \(i\)'s objective performance is
\begin{equation}
\label{eq:true_objective_i}
\Big\langle \psi_i^{\text{obj}}, \bar u_{i,T} \Big\rangle,
\end{equation}
where \(\bar u_{i,T}\) is the time-averaged payoff from \eqref{eq:avg_vector_payoff_i}.

During learning, however, player \(i\) need not use \(\psi_i^{\text{obj}}\) inside its update rule. Instead, at each round \(t \in [T]\), the player may deploy any weight vector
$\psi_{i,t} \in K_i^* \setminus \{0\}$
and use it to guide action selection at round \(t\). 
The deployed weight affects the player's updates, which affect the player's actions, which in turn can influence how the other players respond.
\begin{remark}
\label{rem:responsive_environment}
In our formulation, the remaining players are not assumed to form a fully adversarial environment. Rather, they are modeled as payoff-responsive to the realized interaction; for example, they may update their own actions through no-regret learning based on their own scalar feedback. Consequently, the realized trajectory of play, and hence the average payoff \(\bar u_{i,T}\), depends on the deployed weight used by player \(i\).
\end{remark}

To make this dependence explicit, let \(\bar u_{i,T}(\psi_i)\) denote the average payoff generated when player \(i\) fixes the deployed weight \(\psi_i\) over the full time horizon while the other players continue to adapt. The resulting objective value is then
\begin{equation}
\label{eq:psi_to_value_mapping_i}
\Big\langle \psi_i^{\text{obj}}, \bar u_{i,T}(\psi_i) \Big\rangle.
\end{equation}
Because the other players react to the realized interaction, this mapping need not be maximized at \(\psi_i = \psi_i^{\text{obj}}\). In other words, the weight vector that is best for evaluating outcomes need not be the weight vector that is optimal to deploy inside the learning dynamics. 

This setup leads to the following problem statements.

\begin{problem}
\label{prob:intent_algorithm}
Design an algorithm for player \(i\) that selects actions \(\{a_{i,t}\}_{t=1}^T\) and adaptively chooses deployed weights \(\{\psi_{i,t}\}_{t=1}^T\) to maximize
$\mathbb{E}\left[\left\langle \psi_i^{\text{obj}}, \bar u_{i,T} \right\rangle\right]$,
where \(\psi_i^{\text{obj}}\) is fixed. 
\end{problem}

\begin{problem}
\label{prob:regret_theorem}
Provide a finite-time regret bound for the proposed algorithm.
\end{problem}

\begin{problem}
\label{prob:experiments}
Empirically evaluate 
how adaptively chosen weights \(\{\psi_{i,t}\}_{t=1}^T\) 
lead to improvement in 
$\left\langle \psi_i^{\text{obj}}, \bar u_{i,T} \right\rangle$
relative to a baseline that deploys \(\psi_{i,t} \equiv \psi_i^{\text{obj}}\) for all \(t \in [T]\).
\end{problem}

%% file: sec/4_algorithms.tex
\section{Algorithms}
\label{sec:algorithms}

This section solves Problem~\ref{prob:intent_algorithm}. 
We fix a focal player with index~$i \in N$, and
we introduce a bi-level learning algorithm for player~$i$ that deploys a scalarization \(\psi_i\) to shape its learning dynamics and update its mixed actions from bandit feedback. The following assumptions are
first made. 

\begin{assumption}
\label{ass:standing_algorithms}
Fix a focal player $i \in N$. Then: 
\begin{enumerate}
\item \textbf{Bounded vector payoffs.}
\label{ass:standing_1}
The stage payoff map \(u_i:~\mathcal{A}_i\times\mathcal{A}_{-i}\to\mathbb{R}^{d_i}\)  is coordinate-wise 
bounded, i.e., there exists \(U>0\) such that
\begin{equation}
\label{eq:bounded_payoffs}
\begin{aligned}
\|u_i(a_{i,t},a_{-i,t})\|_\infty \le U \\
\end{aligned}
\end{equation}
for all $(a_{i,t},a_{-i,t})\in\mathcal{A}_i\times\mathcal{A}_{-i}$.

\item \textbf{Polyhedral preference cone.}
\label{ass:standing_2}
The preference cone \(K_i\) is polyhedral. 
\item \textbf{Normalization of candidate weights.}
\label{ass:standing_3}
We restrict attention to normalized candidate weights, that is,
\begin{equation}
\label{eq:psi_unit_norm}
\psi_i \in \bigl(K_i^* \setminus \{0\}\bigr)\cap S^{d_i-1}.
\end{equation}
\end{enumerate}
\end{assumption}

The following proposition shows that if the preference cone is polyhedral, then its dual cone admits a finite generating set.

\begin{proposition}[Finite generating set]
\label{prop:finite_candidate_class}
Suppose Assumption~\ref{ass:standing_algorithms}.\ref{ass:standing_2} holds. Then there exists a finite set
\[
\Psi_i=\{\psi_i^1,\dots,\psi_i^{m}\}\subseteq K_i^*\setminus\{0\}
\]
such that
$K_i^*=\cone(\Psi_i)$.
\end{proposition}

\begin{proof}
Since \(K_i\) is polyhedral, its dual cone \(K_i^*\) is also polyhedral. By the Minkowski--Weyl theorem~\cite{ROC15}, every polyhedral cone is finitely generated. Hence there exists a finite set
$\Psi_i=\{\psi_i^1,\dots,\psi_i^{m}\}\subseteq K_i^*\setminus\{0\}$
such that
$K_i^*=\cone(\Psi_i)$.
\end{proof}

Assumption~\ref{ass:standing_algorithms}.\ref{ass:standing_2} and Proposition~\ref{prop:finite_candidate_class} imply that the algorithm we develop can be restricted to a finite candidate class of scalarizations obtained by normalizing the finite generating set of \(K_i^*\).

The following proposition shows that every scalarization induced by a dual-cone vector is order-preserving with respect to the cone-induced preference relation.

\begin{proposition}[Monotonicity of scalarizations under the cone order] \label{prop:cone_order_scalar_monotone}
Consider the game \(G\) from~\eqref{eq:Vector_valued_Game}, and fix a player \(i \in N\). Let \(K_i \subseteq \mathbb{R}^{d_i}\) be player \(i\)'s preference cone from \eqref{eq:Kidef}, and let \(K_i^*\) be its dual cone. If \(a,b\in\mathbb{R}^{d_i}\) satisfy \(a \le_{K_i} b\), then for every \(\psi_i \in K_i^*\) we have
$\langle \psi_i, a\rangle \le \langle \psi_i, b\rangle$.
Equivalently, for every \(\psi_i \in K_i^*\), the map \(z \mapsto \langle \psi_i, z\rangle\) is order-preserving with respect to \(\le_{K_i}\).
\end{proposition}

\begin{proof}
If \(a \le_{K_i} b\), then \(b-a \in K_i\). For any \(\psi\in ~K_i^*\), by definition of the dual cone,
\(\langle \psi, b-a\rangle \ge 0\), i.e., \(\langle \psi, b\rangle \ge ~\langle \psi, a\rangle\).
\end{proof}


\begin{remark}
We focus on a fixed focal player \(i \in N\) and, throughout the rest of this section, we drop the subscript \(i\) from all quantities associated with that player. The remaining players are not modeled explicitly and may update their behavior arbitrarily, provided the resulting environment remains responsive rather than fully adversarial,
as described in Remark~\ref{rem:responsive_environment}.
\end{remark}

We next introduce the algorithm we use.

\subsection{Block Protocol}

We use a bi-level learning architecture. 
The outer layer operates on a slower timescale, and it selects 
a weight vector \(\psi\in~\Psi\) to deploy in the inner
layer, where \(\Psi~=~\{~\psi^1,~\dots,~\psi^m\}\) is the finite candidate set from Proposition~\ref{prop:cone_order_scalar_monotone}.
The inner layer operates on a faster timescale, and it 
selects actions using the currently deployed scalarization. For each candidate \(\psi^j \in \Psi\), the player maintains an associated mixed action distribution.
Collecting these 
mixed action distributions
gives a matrix of policies \(Q\in\Delta(\mathcal{A})^m\subset\mathbb{R}^{m\times|\mathcal{A}|}\). 
We use~$Q[j, \cdot]$ to denote the~$j^{th}$ row of~$Q$.
For each~$j \in [m]$, 
the row \(Q[j,\cdot]\in~\Delta(\mathcal{A})\) 
is the policy used when deploying \(\psi^j\). The outer layer maintains a distribution over candidate scalarizations,
and this distribution is an element of~\(\Delta(\Psi)\). 
The player changes its policies based on the scalar
feedback it observes, and we use~$Q[j, \cdot]_t$ to denote
the player's policy for~$\psi^j$ at round~$t$. 

The outer layer operates in blocks of rounds of the game,
in the sense that the player fixes 
a vector
\(\psi\) and takes actions according to the fixed scalarized payoff via \(\psi\) 
for a fixed number of rounds. Fix a block length \(L\in\mathbb{N}\) and let
\(h:=~\lceil T/L\rceil\) denote the number of blocks. 
For each block \(k\in[h]\), define
\[
s_0(k):=(k-1)L+1,\qquad s_1(k):=\min\{kL,T\},
\]
\[
\mathcal{I}_k:=[s_0(k),s_1(k)]\cap\mathbb{Z},
\]
where~$s_0(k)$ is the starting index of block~$k$, ~$s_1(k)$ is the last index of block~$k$, and $\mathcal{I}_k$ is the set of indices $t$ of the rounds within block~$k$.
That is, the first block consists of rounds 
with indices~$t \in \{s_0(1), s_0(1) + 1, \ldots, s_1(1)\}$,
the second block consists of rounds
with indices $t \in \{s_0(2), s_0(2) + 1, \ldots, s_1(2)\},$ etc.

At the start of block \(k\), the player's outer loop samples \(J_k\sim p_k\),
where~$p_k \in \Delta(\Psi)$ denotes the outer loop's distribution over~$\Psi$
during block~$k$. 
Then, its
inner loop 
deploys \(\psi^{J_k}\) for all \(t\in\mathcal{I}_k\). 
The player samples actions \(a_t\sim Q[J_k,\cdot]_t\), 
where~$Q[J_k,\cdot]_t$ is 
the distribution at round $t$ over actions associated with the candidate scalarization with index~$J_k$.
It updates only the  row \(Q[J_k,\cdot]_t\) using the shaping reward
\begin{equation}
\label{eq:inner_shaping_reward}
r_t^{(J_k)}:=\langle \psi^{J_k},u_t\rangle,
\qquad t\in\mathcal{I}_k.
\end{equation}
This update is performed
as follows. 
We form the IX-stabilized gradient estimate\footnote{
\ifbool{techreport}{
  Reviews of IX-stabilized gradient estimates and online mirror descent are provided in Appendix A.
}{
  Reviews of IX-stabilized gradient estimates and online mirror descent are provided in~\cite[Appendix A]{techreport}.
}
}
\(\widehat g_t\) (see~\cite{Koc14,NG15,LC20,BC12}) via 
\begin{equation}
\label{eq:inner_ix_estimator}
\widehat g_t[a]
:=
-\frac{\mathbb{I}\{a_t=a\}}{Q[J_k,\cdot](a_t)+\gamma_q}\, r_t^{(J_k)},
\qquad a\in\mathcal A,
\end{equation}
where \(\mathbb{I}\{\cdot\}\) denotes the indicator function, and \(\gamma_q>0\) is the implicit-exploration parameter used to stabilize the importance-weighted estimator. We then apply one 
online mirror descent
(OMD) step~\cite{shalev-shwartz_omd_hazan_bubeck} to the active policy row, 
and it takes the form
\begin{equation}
\label{eq:inner_omd_update_block}
Q[J_k,\cdot]_{t+1}
\in
\argmin_{q\in\Delta(\mathcal A)}
\left\{
\langle q,\widehat g_t\rangle + \tfrac{1}{\eta_q} D_{R_q}\big(q\|Q[J_k,\cdot]_t\big)
\right\}, 
\end{equation}
where \(D_{R_q}(\cdot\|\cdot)\) denotes the Bregman divergence associated with the regularizer \(R_q\) on the simplex \(\Delta(\mathcal A)\). 
All rows \(Q[j,\cdot]\) for \(j\neq J_k\) are left unchanged during
time block \(k\).

The outer layer is evaluated using the objective weight-vector \(\psi^{\text{obj}}\) 
from~\eqref{eq:obj_psi}
via the block-average reward
\begin{equation}
\label{eq:block_objective_reward}
r_k^{{\text{obj}}} := \frac{1}{|\mathcal{I}_k|}\sum_{t\in\mathcal{I}_k}\langle \psi^{{\text{obj}}},u_t\rangle.
\end{equation}
Using the IX estimator on \(\Delta([m])\) with the preference vector index~$J_k$
and feedback \(r_k^{{\text{obj}}}\), we obtain \(\widehat g_k^{(P)}\) via
\begin{equation}
\label{eq:outer_ix_estimator}
\widehat g_k^{(P)}[j]
:=
-\frac{\mathbb{I}\{J_k=j\}}{p_k(j)+\gamma_p}\, r_k^{{\text{obj}}},
\qquad j\in[m],
\end{equation}
where \(\gamma_p>0\) is the implicit-exploration parameter used to stabilize the importance-weighted estimator. 
The quantity \(\widehat g_k^{(P)}\) is the estimated loss vector for the outer learner, and it converts the scalar feedback \(r_k^{{\text{obj}}}\) from the sampled preference index into a full vector that is used in the OMD update and outer regret analysis.
We then apply one OMD step to update \(p_k\) using
\begin{equation}
\label{eq:outer_omd_update_block}
p_{k+1}
\in
\argmin_{p\in\Delta([m])}
\left\{
\langle p,\widehat g_k^{(P)}\rangle + \tfrac{1}{\eta_p} D_{R_p}(p\|p_k)
\right\},
\end{equation}
 where \(D_{R_p}(\cdot\|\cdot)\) denotes the Bregman divergence associated with the regularizer \(R_p\) on the simplex \(\Delta(\Psi)\).

\FloatBarrier

\begin{algorithm}[H]
\caption{\textsc{OuterAlg}}
\label{alg:outer_weight}
\begin{algorithmic}[1]
\Require 
Candidate scalarizations \(\Psi=\{\psi^1,\dots,\psi^m\}\), time horizon \(T\), block size \(L\), objective weight vector \(\psi^{{\text{obj}}}\), step sizes \(\eta_p,\eta_q>0\), implicit-exploration parameters \(\gamma_p,\gamma_q\ge 0\), and regularizers \(R_p,R_q\)

\State \(h\gets \lceil T/L\rceil\); initialize \(p_1\in\Delta([m])\), \(Q\in\Delta(\mathcal A)^m\) with no zero entries
\For{$k=1,\dots,h$}
  \State \(\mathcal I_k \gets [(k-1)L+1,\min\{kL,T\}]\cap\mathbb Z\)
  \State sample \(J_k\sim p_k\)
  \State \((Q,\ r_k^{{\text{obj}}})\gets \textsc{InnerAlg}(\)
\Statex \hspace{2.9cm}\(J_k,\psi^{J_k},\psi^{{\text{obj}}},\mathcal I_k,Q,\eta_q,\gamma_q,R_q)\)
  \State Compute \(\widehat g_k^{(P)}\) with \eqref{eq:outer_ix_estimator} using \((p_k,J_k,r_k^{{\text{obj}}},\gamma_p)\)
  \State Compute \(p_{k+1}\) using \eqref{eq:outer_omd_update_block} 
\EndFor
\State \Return \(\{p_k\}_{k=1}^{h+1}\), \(Q\)
\end{algorithmic}
\end{algorithm}


\begin{algorithm}[H]
\caption{\textsc{InnerAlg}}
\label{alg:inner_action}
\begin{algorithmic}[1]
\Require Active index \(J\in[m]\), deployed scalarization \(\psi^{J}\), objective weight-vector \(\psi^{{\text{obj}}}\), 
block index \(\mathcal I\subseteq[T]\), matrix \(Q~\in~\Delta(\mathcal A)^m\), step size \(\eta_q>0\), implicit-exploration parameter \(\gamma_q\ge 0\), and regularizer \(R_q\)
\State \(S^{{\text{obj}}}\gets 0\)
\For{$t\in\mathcal I$}
  \State sample \(a_t\sim Q[J,\cdot]\); observe \(u_t\) from~\eqref{eq:stage_payoff_i}  
  \State \(S^{{\text{obj}}}\gets S^{{\text{obj}}}+\langle \psi^{{\text{obj}}},u_t\rangle\)
  \State \(r_t\gets \langle \psi^{J},u_t\rangle\)
  \State Compute \(\widehat g_t\) by \eqref{eq:inner_ix_estimator} using \((Q[J,\cdot],a_t,r_t,\gamma_q)\)
  \State Compute \(Q[J,\cdot]_{t+1}\) using~\eqref{eq:inner_omd_update_block} 
\EndFor
\State \(r^{{\text{obj}}}\gets S^{{\text{obj}}}/|\mathcal I|\)
\State \Return \(Q\), \(r^{{\text{obj}}}\)
\end{algorithmic}
\end{algorithm}

\FloatBarrier

Algorithms~\ref{alg:outer_weight} and~\ref{alg:inner_action} implement the block protocol described above. Algorithm~\ref{alg:outer_weight} updates the block-level distribution \(p_k\in~\Delta([m])\) using the block feedback \(r_k^{\text{obj}}\), and calls Algorithm~\ref{alg:inner_action}, which runs within the block and updates only the active row \(Q[J_k,\cdot] \in \Delta(\mathcal A)\) using shaping rewards induced by the deployed scalarization \(\psi^{J_k}\).


%% file: sec/5_regretbounds.tex
\section{Regret Analysis}
\label{sec:theory}
This section solves Problem~\ref{prob:regret_theorem} and
derives finite-time regret bounds for the bi-level protocol
in Algorithms~\ref{alg:outer_weight} and~\ref{alg:inner_action}. 

\subsection{Inner-Layer Regret}
\label{subsec:inner_regret_bounds}



Fix a block \(k\) and its active index \(J_k\in[m]\). Conditioned on choosing \(J_k\), the inner algorithm
faces a sequence of bandit linear losses over the simplex \(\Delta(\mathcal A)\) induced by the deployed scalarization \(\psi^{J_k}\). 
We compare the inner algorithm to the best fixed mixed action in hindsight, namely 
\begin{equation}
\label{eq:inner_block_best_hindsight}
q_k^\star \in \argmin_{q\in\Delta(\mathcal A)}
\sum_{t\in\mathcal I_k} \big\langle q,\ \widehat g_t\big\rangle,
\end{equation}
i.e., the single distribution that would have minimized the accumulated estimated loss if it had been played throughout the block. The following bound controls the inner algorithm's performance relative to this hindsight benchmark.


\begin{theorem}[Inner block regret bound]
\label{thm:inner_block_regret}
Suppose Assumption~\ref{ass:standing_algorithms} holds, fix \(\gamma_q>0\), and suppose Algorithm~2 uses the negative entropy regularizer
$R_q(q)~=~\sum_{a\in\mathcal A} q(a)\log q(a)$
for the inner OMD update. Then, for the hindsight-optimal comparator \(q_k^\star\) in \eqref{eq:inner_block_best_hindsight}, and with \(T_k:=|\mathcal I_k|\), we have
\begin{equation}
\label{eq:inner_block_regret}
\sum_{t\in\mathcal I_k}\big\langle Q_t[J_k,\cdot]-q_k^\star,\widehat g_t\big\rangle
\le
\frac{\sqrt{d}\,U}{\gamma_q}\sqrt{2T_k\log|\mathcal A|}.
\end{equation}
\end{theorem}

\begin{proof}
\ifbool{techreport}{
  See Appendix~\ref{app:inner_block_regret}.
}{
  See~\cite[Appendix B]{techreport}.
}
\end{proof}

Theorem~\ref{thm:inner_block_regret} shows that, over each block, the inner algorithm incurs only sublinear regret relative to the best fixed mixed action chosen in hindsight for the deployed scalarization in that block. Thus, when the outer layer fixes a scalarization, the inner layer can learn to act nearly as well as the best fixed action rule for that scalarized objective.

\subsection{Outer-Layer Regret}
\label{subsec:outer_regret_bounds}

The outer update 
in Algorithm~\ref{alg:outer_weight}
runs over \(\Delta([m])\) for \(h\) rounds using the block-average feedback
\[
r_k^\star := \frac{1}{|\mathcal I_k|}\sum_{t\in\mathcal I_k}\langle \psi^\star,u_t\rangle,
\]
the IX estimator in~\eqref{eq:outer_ix_estimator} with parameters \((q_t,a_t,r_t,\gamma)=~(p_k,J_k,r_k^\star,\gamma_p)\), and one OMD step in~\eqref{eq:outer_omd_update_block} with mirror map \(R_p\).
The following result bounds the regret of the outer algorithm,
namely Algorithm~\ref{alg:outer_weight}, under these conditions. 


\begin{theorem}[Outer block regret bound]
\label{thm:outer_block_regret}
Suppose Assumption~\ref{ass:standing_algorithms} holds, fix \(\gamma_p>0\), and suppose Algorithm~1 uses the negative entropy regularizer
$R_p(p)=\sum_{j=1}^m p(j)\log p(j)$
for the outer OMD update. Then, for any comparator \(p\in\Delta([m])\),
\begin{equation}
\label{eq:outer_regret_general}
\sum_{k=1}^h \big\langle p_k-p,\widehat g_k^{(P)}\big\rangle
\le
\frac{\sqrt{d}\,U}{\gamma_p}\sqrt{2h\log m}.
\end{equation}
\end{theorem}

\begin{proof}
\ifbool{techreport}{
  See Appendix~\ref{app:outer_block_regret}.
}{
  See~\cite[Appendix C]{techreport}.
}
\end{proof}

Theorem~\ref{thm:outer_block_regret} shows that the outer algorithm incurs only sublinear regret over blocks relative to any fixed distribution over candidate scalarizations chosen in hindsight. Thus, at the slower timescale, the outer layer can learn to select scalarizations nearly as well as the best fixed mixture of candidate weights for the objective feedback observed across blocks.

\subsection{Regret of the Bilevel Algorithm}
\label{subsec:Bilevel_regret_bounds}

We define the regret of the bilevel algorithm by
\begin{equation}
\label{eq:bilevel_regret_def}
\mathcal R_T^{\mathrm{bi}}
:=
\sum_{k=1}^h \big\langle p_k-p^\star,\widehat g_k^{(P)}\big\rangle
+
\sum_{k=1}^h \sum_{t\in\mathcal I_k}\big\langle Q_t[J_k,\cdot]-q_k^\star,\widehat g_t\big\rangle.
\end{equation}
Here, \(p^\star\) is the best fixed distribution over candidate scalarizations chosen in hindsight, and \(q_k^\star\) is the best fixed mixed action chosen in hindsight on block \(\mathcal I_k\) for the scalarization deployed on that block. The following theorem establishes the regret bound for the full bilevel algorithm.

\begin{theorem}[Regret of the bilevel algorithm]
\label{thm:bilevel_regret}
Suppose all conditions of Theorems~\ref{thm:inner_block_regret}
and~\ref{thm:outer_block_regret} hold. Then the regret of the bilevel algorithm satisfies
\begin{equation}
\label{eq:bilevel_regret_bound}
\mathcal R_T^{\mathrm{bi}}
\le
\sqrt d\,U\left(
\frac{\sqrt{2h\log m}}{\gamma_p}
+
\frac{\sqrt{2hT\log|\mathcal A|}}{\gamma_q}
\right).
\end{equation}
\end{theorem}



\begin{proof}
\ifbool{techreport}{
  See Appendix~\ref{app:bilevel_regret}.
}{
  See~\cite[Appendix D]{techreport}.
}
\end{proof}

Theorem~\ref{thm:bilevel_regret} states that the total error of the bilevel procedure decomposes into two sublinear components: an outer regret term that measures how well the algorithm chooses among deployed scalarizations across blocks, and an inner regret term that measures how well the algorithm learns actions within each deployed block. The overall procedure performs nearly as well as a benchmark that, up to sublinear error, combines the best fixed outer weighting strategy with the best fixed within-block action choices under the deployed scalarizations. In this sense, the theorem formalizes the
fact that the bilevel architecture can simultaneously learn which scalarizations are useful and 
learn how to act under them.




%% file: sec/6_Empirical_Validation.tex
\section{Simulation Results}
\label{sec:experiments}

This section presents simulation results for a vector-valued
extension of the classic Bach or Stravinsky game~\cite{OR94}. 
To demonstrate the performance of Algorithms~\ref{alg:outer_weight} and~\ref{alg:inner_action} we compare with an exponential weights algorithm with implicit exploration (Exp-IX) from~\cite{NG15}.

\subsection{Problem Setup}
We consider a repeated two-action game with actions \(\{B,S\}\) for both players and a \(4\)-dimensional vector payoff \(u(a,b)\in\mathbb{R}^4\) defined as 
\[
\begin{array}{c|cc}
 & B & S \\ \hline
B & (1,1,1,0) & (-1,1,1,-1) \\
S & (1,-1,-1,1) & (0,1,1,1)
\end{array}
\]
At round \(t\), the realized outcome is \(u_t=u(a_t,b_t)\). The focal player is evaluated using
\[
\psi^{{\text{obj}}}=\Bigl(\tfrac{\sqrt2}{2},\tfrac{\sqrt2}{2},0,0\Bigr),
\]
while the opponent updates using the weight vector
\[
\phi^{{\text{obj}}}=\Bigl(0,0,\tfrac{\sqrt2}{2},\tfrac{\sqrt2}{2}\Bigr).
\]
Thus, the focal evaluation signal is \(\langle \psi^{{\text{obj}}},u_t\rangle\), while the opponent learns from \(\langle \phi^{{\text{obj}}},u_t\rangle\).

Under these fixed choices of \(\psi^{{\text{obj}}}\) and \(\phi^{{\text{obj}}}\), the induced scalar payoff matrix is
\[
\begin{array}{c|cc}
 & B & S \\ \hline
B & (\sqrt2,\tfrac{\sqrt2}{2}) & (0,0) \\
S & (0,0) & (\tfrac{\sqrt2}{2},\sqrt2)
\end{array}
\]
which has two pure Nash equilibria, i.e., 
\(BB\) and \(SS\). Deployed weights can alter the induced incentives. For example, deploying
\[
\psi'=\Bigl(\tfrac{\sqrt2}{2},0,0,-\tfrac{\sqrt2}{2}\Bigr)
\]
inside the focal player's updates, while evaluation remains along \(\psi^{{\text{obj}}}\),
yields the scalar game
\[
\begin{array}{c|cc}
 & B & S \\ \hline
B & (\tfrac{\sqrt2}{2},\tfrac{\sqrt2}{2}) & (0,0) \\
S & (0,0) & (-\tfrac{\sqrt2}{2},\sqrt2)
\end{array}
\]
in which \(B\) is strictly dominant for the focal player and \(SS\) is no longer an equilibrium.


\subsection{Results and Comparisons to Baselines}
We evaluate our bi-level framework on the four dimensional game described above with $1,000$ independent runs that each consist of $10,000$ rounds.
We compare two scenarios. In Scenario 1 (Exp-IX vs. Exp-IX) both players use standard exponential weights with implicit exploration (Exp-IX)~\cite{NG15}, and in Scenario 2 (Bi-level vs Exp-IX) one player uses our bi-level architecture while the other uses Exp-IX.
The bi-level player maintains a distribution over 
\begin{multline}
\Psi = \Big\{(\tfrac{\sqrt2}{2},\tfrac{\sqrt2}{2},0,0),(\tfrac12,\tfrac12,\tfrac12,\tfrac12),
(\tfrac12,\tfrac12,-\tfrac12,-\tfrac12)\Big\}.
\end{multline}
The choice of a weight vector~$\psi \in \Psi$ 
is performed every $L=500$ rounds using Algorithm~\ref{alg:outer_weight} with a learning rate of $\eta_p=0.1$. 
The inner algorithm, Algorithm~\ref{alg:inner_action}, uses $\eta_q=0.1$. 
For both algorithms, we use an implicit exploration parameter of $\gamma_p=\gamma_q=0.2.$
In both scenarios, the baseline Exp-IX players use  $\gamma=0.2.$ 

Figure~\ref{fig:toy4d_bos_results} presents the distribution of equilibrium outcomes over the $1,000$ runs for both scenarios.
Each run's equilibrium is determined based on the majority action profile over the final $1,000$ steps.
Because the learning dynamics are stochastic, the realized action profile in a single run does not necessarily settle on the same outcome for every round.
We therefore classify each run by the equilibrium whose action profile appears most often over the final $1,000$ steps.
In Scenario 1 (Exp-IX vs. Exp-IX), convergence to the $BB$ and $SS$ equilibria 
happens with nearly the same frequency: 
53\% of runs end at~$BB$, while 47\% end at \(SS\).

In Scenario 2 (Bi-level vs Exp-IX) there is a strong asymmetry: $82\%$ of runs converge to $BB,$ $9\%$ of runs converge to $SS,$ and $9\%$ converge to neither.
This shift between the scenarios demonstrates that the bi-level player's adaptive scalarization successfully biases the coupled learning dynamics toward the $BB$ equilibrium, which yields a higher payoff under its true objective weight vector $\psi^{{\text{obj}}}~=~(\sqrt{2}/2,{\sqrt2}/2,0,0).$
The increase in non-convergent outcomes from 0\% to 9\% suggests that adaptive scalarization can sometimes delay convergence, though the overall benefit of steering toward the preferred equilibrium may outweigh this drawback.

Figure~\ref{fig:scenario_2} shows the mean scalarized reward trajectories for runs that converged to $BB$ in the top plot, and 
it shows the same quantity for 
runs that converged to $SS$ in the bottom plot, with the shaded regions corresponding to $\pm 1$ standard deviation. These trajectories are smoothed with a moving average over $300$ rounds. 
The gold and blue lines depict $\langle\psi^{{\text{obj}}},u_t\rangle$ and $\langle \phi^{{\text{obj}}},u_t\rangle,$ respectively.
The overall trend is that the bi-level algorithm earns a higher reward on trajectories that end in $BB$ and a lower reward on trajectories that end in $SS.$  
These trajectories demonstrate that adaptive scalarization enables the bi-level player to bias the coupled learning dynamics toward equilibria that are more favorable under its objective weight vector.

These results collectively illustrate the core mechanism behind our approach: in a payoff-responsive environment, the interaction trajectory depends on how the focal player evaluates outcomes inside its learning updates. By adapting the deployed scalarization, the focal player changes its own action updates, which in turn changes the opponent’s response and the induced equilibrium that emerges. As a result, the bi-level method can change its weights to steer play toward trajectories that yield higher payoff values under the fixed objective weight \(\psi^{{\text{obj}}}\).

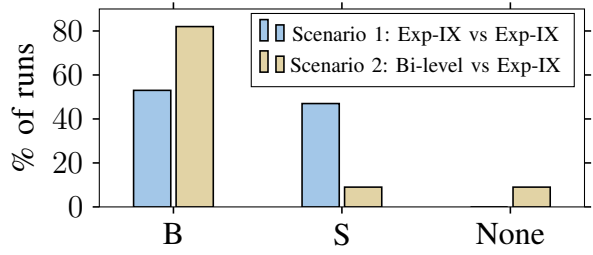
\begin{figure}
    \centering
    \input{figures/BoS}
    \caption{Outcomes over \(1000\) runs (\(T=10^4\)). Scenario 1: Exp-IX vs.\ Exp-IX under fixed weights. Scenario 2: bi-level (Algorithms~\ref{alg:outer_weight}--\ref{alg:inner_action}) vs.\ Exp-IX.
    We see that when both players use Exp-IX, the runs split nearly evenly between the $BB$ and $SS$ equilibria. When the focal player instead uses our bi-level method, the outcome distribution shifts strongly towards $BB,$ showing that adaptive scalarization biases the learning dynamics toward the equilibrium preferred by the focal player.
    }
    \label{fig:toy4d_bos_results}
\end{figure}

\setlength{\figwidth}{8cm}
\setlength{\figheight}{4.2cm}


\begin{figure}
    \centering
    \input{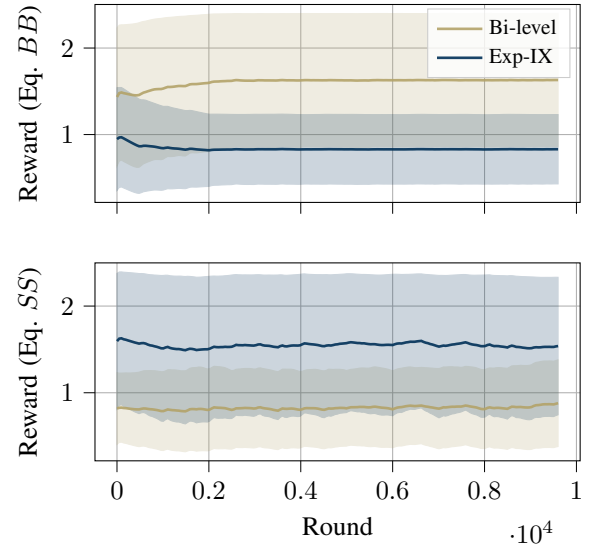}
    \caption{Scalarized reward trajectories for Scenario 2 (Bi-level vs Exp-IX) by equilibrium type. The top plot shows $BB$ equilibrium runs and the bottom plot shows $SS$ equilibrium runs. The results were averaged over $1000$ runs and the shaded regions correspond to $\pm1$ standard deviation.    
    We see that among runs that converge to $BB,$ the bi-level player receives consistently higher reward than Exp-IX over the course of learning.
    Among the runs that converge to $SS,$ the bi-level player attains lower reward, which is consistent with $SS$ being less favorable under the focal player's objective.
    }
    \label{fig:scenario_2}
\end{figure}

%% file: figures/BoS.tex
\begin{tikzpicture}
\begin{axis}[
    ybar,
    ymin=0, ymax=90,
    ylabel={\% of runs},
    width=8.0cm,
    height=4.2cm,
    axis lines=box,
    tick align=outside,
    tick style={black, line width=0.6pt},
    major tick length=2.5pt,
    symbolic x coords={B,S,None},
    xtick=data,
    xticklabel style={font=\large},
    yticklabel style={font=\large},
    ylabel style={font=\large},
    bar width=14pt,
    enlarge x limits=0.22,
    legend style={
        at={(0.98,0.98)},
        anchor=north east,
        font=\footnotesize,
        draw=black,
        fill=white
    },
]

\addplot[
    fill=pastelBlue,
    draw=black,
    line width=0.6pt
] coordinates {
    (B,53)
    (S,47)
    (None,0)
};

\addplot[
    fill=pastelSand,
    draw=black,
    line width=0.6pt
] coordinates {
    (B,82)
    (S,9)
    (None,9)
};

\legend{Scenario 1: Exp-IX vs Exp-IX, Scenario 2: Bi-level vs Exp-IX}

\end{axis}
\end{tikzpicture}






%% file: sec/7_Conclusion.tex
\section{Conclusion}
\label{sec:conclusion}

We studied repeated vector-valued games where evaluation uses a fixed objective weight, but the player may deploy alternative weights inside its updates to influence coupled adaptation. We proposed a bi-level algorithm with finite set of candidate weights and established finite-time regret bounds, along with experiments showing that adaptive deployment can steer outcomes toward a higher objective value than fixed-weight baselines. Future work will explore continuous weight selection and 
seek to reduce regret even further. 

%% file: app/app_A.tex
\appendix

This appendix collects the online mirror descent material needed for the regret analysis in Appendices~\ref{app:inner_block_regret} and~\ref{app:outer_block_regret}.

\subsection{Auxiliary Definitions and Regret Bound} \label{OMD_review}

We first present two basic definitions. 

\begin{definition}[Bregman divergence]
\label{def:bregman_div_app}
Let \(\mathcal{K}\subseteq\mathbb{R}^p\) be convex and let \(R:\mathcal{K}\to\mathbb{R}\) be differentiable. The Bregman divergence induced by \(R\) is
\begin{equation}
\label{eq:bregman_div_app}
D_R(z\|w)
:=
R(z)-R(w)-\langle \nabla R(w), z-w\rangle.
\end{equation}
\end{definition}

\begin{definition}[Online mirror descent]
\label{def:omd_app}
Let \(\mathcal{K}\subseteq\mathbb{R}^p\) be a nonempty compact convex set, and let \(R:\mathcal{K}\to\mathbb{R}\) be a differentiable regularizer. At each round \(t\in[T]\), a player chooses \(z_t\in\mathcal{K}\) and observes a subgradient \(g_t\in\mathbb{R}^p\) of the current loss. Given a step size \(\eta>0\), the OMD update is
\begin{equation}
\label{eq:omd_update_app}
z_{t+1}\in \argmin_{z\in\mathcal{K}}
\left\{
\langle z,g_t\rangle + \frac{1}{\eta}D_R(z\|z_t)
\right\}.
\end{equation}
\end{definition}

The following lemma states a standard regret bound. 

\begin{lemma}[Standard OMD regret bound]
\label{lem:omd_regret_app}
Consider the setting of Definition~\ref{def:omd_app}. Suppose \(R\) is \(\rho\)-strongly convex on \(\mathcal{K}\) with respect to a norm \(\|\cdot\|\), and let \(\|\cdot\|_*\) denote the corresponding dual norm. Define the Bregman diameter
\[
\mathsf{R}^2 := \sup_{z,w\in\mathcal{K}} D_R(z\|w),
\]
which is finite since \(\mathcal{K}\) is compact and \(D_R\) is continuous. If \(\|g_t\|_*\le \mathsf{L}\) for all \(t\), then for any fixed comparator \(z\in\mathcal{K}\),
\begin{equation}
\label{eq:omd_regret_general_app}
\sum_{t=1}^T \langle z_t-z, g_t\rangle
\le
\frac{\mathsf{R}^2}{\eta}
+
\frac{\eta}{2\rho}\sum_{t=1}^T \|g_t\|_*^2
\le
\frac{\mathsf{R}^2}{\eta}
+
\frac{\eta T \mathsf{L}^2}{2\rho}.
\end{equation}
In particular, choosing
\[
\eta=\frac{\mathsf{R}}{\mathsf{L}}\sqrt{\frac{2\rho}{T}}
\]
yields
\begin{equation}
\label{eq:omd_regret_simplified_app}
\sum_{t=1}^T \langle z_t-z, g_t\rangle
\le
\mathsf{R}\,\mathsf{L}\sqrt{\frac{2T}{\rho}}.
\end{equation}
\end{lemma}

For the proofs of Theorems~\ref{thm:inner_block_regret} and~\ref{thm:outer_block_regret}, we apply Lemma~\ref{lem:omd_regret_app} on probability simplices with the negative-entropy regularizer
\[
R(x)=\sum_j x(j)\log x(j).
\]
In this case, \(R\) is \(1\)-strongly convex with respect to \(\|\cdot\|_1\).
Then \(\rho=1\) and the dual norm is \(\|\cdot\|_\infty\). Moreover, the associated Bregman divergence is the Kullback--Leibler divergence.

For the simplex \(\Delta(\mathcal A)\), the Bregman diameter satisfies
\[
\mathsf{R}^2 \le \log |\mathcal A|,
\]
and for the simplex \(\Delta([m])\) it satisfies
\[
\mathsf{R}^2 \le \log m.
\]

\subsection{Proof of Theorem~\ref{thm:inner_block_regret}}
\label{app:inner_block_regret}

Fix a block \(\mathcal I_k\), and let \(T_k:=|\mathcal I_k|\). The inner update over this block is OMD on \(\Delta(\mathcal A)\) with regularizer
\[
R_q(q)=\sum_{a\in\mathcal A} q(a)\log q(a).
\]
Therefore, by the simplex specialization in Appendix~\ref{OMD_review}, 

Lemma~\ref{lem:omd_regret_app} applies with
\[
\rho=1,
\qquad
\|\cdot\|_*=\|\cdot\|_\infty,
\qquad
\mathsf{R}^2\le \log |\mathcal A|.
\]

It remains to bound \(\|\widehat g_t\|_\infty\). By \eqref{eq:inner_ix_estimator}, the vector \(\widehat g_t\) has only one nonzero entry, namely at the sampled action \(a_t\). Therefore,
\[
\|\widehat g_t\|_\infty
=
\frac{|\langle \psi^{J_k},u_t\rangle|}{Q_t[J_k,\cdot](a_t)+\gamma_q}
\le
\frac{|\langle \psi^{J_k},u_t\rangle|}{\gamma_q}.
\]
Since \(\psi^{J_k}\in S^{d-1}\), we have \(\|\psi^{J_k}\|_2=1\). By Assumption~\ref{ass:standing_algorithms}.\ref{ass:standing_1},
\[
\|u_t\|_\infty \le U,
\]
and hence
\[
\|u_t\|_2 \le \sqrt{d}\,\|u_t\|_\infty \le \sqrt{d}\,U.
\]
Therefore, by Cauchy--Schwarz,
\[
|\langle \psi^{J_k},u_t\rangle|
\le
\|\psi^{J_k}\|_2\,\|u_t\|_2
\le
\sqrt{d}\,U.
\]
Thus,
\[
\|\widehat g_t\|_\infty \le \frac{\sqrt d\,U}{\gamma_q}
\qquad \text{for all } t\in\mathcal I_k.
\]
Then Lemma~\ref{lem:omd_regret_app} applies with
\[
\mathsf{L}=\frac{\sqrt d\,U}{\gamma_q}.
\]

Applying \eqref{eq:omd_regret_simplified_app} with the comparator \(q_k^\star\) from \eqref{eq:inner_block_best_hindsight} over the \(T_k\) rounds in block \(\mathcal I_k\), we obtain
\[
\sum_{t\in\mathcal I_k}\big\langle Q_t[J_k,\cdot]-q_k^\star,\widehat g_t\big\rangle
\le
\sqrt{\log |\mathcal A|}\cdot \frac{\sqrt d\,U}{\gamma_q}\sqrt{2T_k},
\]
which is exactly
\[
\sum_{t\in\mathcal I_k}\big\langle Q_t[J_k,\cdot]-q_k^\star,\widehat g_t\big\rangle
\le
\frac{\sqrt d\,U}{\gamma_q}\sqrt{2T_k\log |\mathcal A|}.
\]
\qed

\subsection{Proof of Theorem~\ref{thm:outer_block_regret}}
\label{app:outer_block_regret}

The outer update in Algorithm~\ref{alg:outer_weight} is OMD on the simplex \(\Delta([m])\) with regularizer
\[
R_p(p)=\sum_{j=1}^m p(j)\log p(j).
\]
Therefore, by the simplex specialization in Appendix~\ref{OMD_review}, 
Lemma~\ref{lem:omd_regret_app} applies with
\[
\rho=1,
\qquad
\|\cdot\|_*=\|\cdot\|_\infty,
\qquad
\mathsf{R}^2\le \log m.
\]

It remains to bound \(\|\widehat g_k^{(P)}\|_\infty\). By \eqref{eq:outer_ix_estimator}, the vector \(\widehat g_k^{(P)}\) has only one nonzero entry, namely at the sampled index \(J_k\). Hence,
\[
\|\widehat g_k^{(P)}\|_\infty
=
\frac{|r_k^{obj}|}{p_k(J_k)+\gamma_p}
\le
\frac{|r_k^{obj}|}{\gamma_p}.
\]
Since the objective weight is normalized, \(\|\psi^{obj}\|_2=1\), and Assumption~\ref{ass:standing_algorithms}.\ref{ass:standing_1} gives
\[
\|u_t\|_\infty \le U,
\]
and therefore we have
\[
\|u_t\|_2 \le \sqrt d\,\|u_t\|_\infty \le \sqrt d\,U.
\]
Thus, by Cauchy--Schwarz,
\[
|\langle \psi^{obj},u_t\rangle|
\le
\|\psi^{obj}\|_2\,\|u_t\|_2
\le
\sqrt d\,U.
\]
Therefore,
\[
|r_k^{obj}| \le \sqrt d\,U
\]
and 
\[
\|\widehat g_k^{(P)}\|_\infty \le \frac{\sqrt d\,U}{\gamma_p}
\qquad \text{for all } k\in[h].
\]

Applying \eqref{eq:omd_regret_simplified_app} over the \(h\) outer rounds with comparator \(p\in\Delta([m])\), we obtain
\[
\sum_{k=1}^h \big\langle p_k-p,\widehat g_k^{(P)}\big\rangle
\le
\sqrt{\log m}\cdot \frac{\sqrt d\,U}{\gamma_p}\sqrt{2h},
\]
which gives
\[
\sum_{k=1}^h \big\langle p_k-p,\widehat g_k^{(P)}\big\rangle
\le
\frac{\sqrt d\,U}{\gamma_p}\sqrt{2h\log m}.
\]
\qed

\subsection{Proof of Theorem~\ref{thm:bilevel_regret}}
\label{app:bilevel_regret}

The regret of the bi-level algorithm in~\eqref{eq:bilevel_regret_def} can be equivalently written as
\begin{multline}
    \mathcal R_T^{\text{bi}} = \max_{\substack{p\in\Delta([m]) \\ q_1,\dots,q_h\in\Delta(A)}}
   \left[ \sum_{k=1}^h \langle p_k-p, \hat g_k^{(P)}\rangle \right. \\ \left.
    +\sum_{k=1}^h\sum_{t_k\in I_k}\langle Q_t[J_k,\cdot]-q_k,\hat g_t\rangle\right],
\end{multline}
which can be bounded with
\begin{align}
    \mathcal R_T &\leq \max_{p\in\Delta([m])}
    \sum_{k=1}^h \langle p_k-p, \hat g_k^{(P)}\rangle \\
    &+\max_{q_1,\dots,q_h\in\Delta(A)}\sum_{k=1}^h\sum_{t_k\in I_k}\langle Q_t[J_k,\cdot]-q_k,\hat g_t\rangle\\
    &=\mathcal R^{out}_T+\sum_{k=1}^h \mathcal{R}_k^{in},\label{eq:decomp_regret}
\end{align}
where
\begin{align}
    \mathcal R^{out}_T &= \max_{p\in\Delta([m])}
    \sum_{k=1}^h \langle p_k-p, \hat g_k^{(P)}\rangle \\
    \mathcal{R}^{in}_k &=\max_{q\in\Delta(A)}\sum_{t_k\in I_k} \langle Q_t[J_k,\cdot]-q,\hat g_t\rangle.
\end{align}
Theorem~\ref{thm:outer_block_regret} gives
\begin{equation}
    \mathcal R^{out}_T\leq \frac{\sqrt d U}{\gamma_p}\sqrt{2h \log m},
\end{equation}
and Theorem~\ref{thm:inner_block_regret} gives
\begin{equation}
    \mathcal{R}^{in}_k\leq \frac{\sqrt{d}\,U}{\gamma_q}\,\sqrt{2|I_k|\log|\mathcal A|}.\label{eq:outer_bound}
\end{equation}
Then, summing over $k\in[h]$ gives 
\begin{align}
    \sum_{k=1}^h \mathcal{R}_k^{in} &= \sum_{k=1}^h \sum_{t_k\in I_k}\langle Q_t[J_k,\cdot]-q^*,\hat g_t\rangle \\
    &\leq \frac{\sqrt{d}\,U}{\gamma_q}\,\sqrt{2\log|\mathcal A|}\sum_{k=1}^h\sqrt{|I_k|}\\
    &\leq \frac{\sqrt{d}\,U}{\gamma_q}\,\sqrt{2hT\log|\mathcal A|}, \label{eq:inner_bound}
\end{align}
where the last line follows from an application of the Cauchy-Schwarz inequality.
Substituting \eqref{eq:outer_bound} and \eqref{eq:inner_bound} in \eqref{eq:decomp_regret} gives
\begin{equation}
    \mathcal R_T \leq \sqrt{d}\,U\left(\frac{\sqrt{2h\log m}}{\gamma_p}  +\frac{\sqrt{2hT \log |A|}}{\gamma_q}\right).
\end{equation}
\qed

%% file: iclr2025_conference.bib
@book{OR94,
  title={A course in game theory},
  author={Osborne, Martin J and Rubinstein, Ariel},
  year={1994},
  publisher={MIT press}
}

@book{ROC15,
  title={Convex analysis},
  author={Rockafellar, R Tyrrell},
  volume={28},
  year={1997},
  publisher={Princeton university press}
}

@book{boyd04,
  title={Convex optimization},
  author={Boyd, Stephen and Vandenberghe, Lieven},
  year={2004},
  publisher={Cambridge university press}
}

@article{BC12,
  title={Regret Analysis of Stochastic and Nonstochastic Multi-armed Bandit Problems},
  author={Bubeck, S{\'e}bastien and Cesa-Bianchi, Nicolo},
  journal={Machine Learning},
  volume={5},
  number={1},
  pages={1--122},
  year={2012}
}

@book{LC20,
  title     = {Bandit Algorithms},
  author    = {Lattimore, Tor and Szepesv{\'a}ri, Csaba},
  publisher = {Cambridge University Press},
  year      = {2020}
}

@inproceedings{NG15,
author = {Neu, Gergely},
title = {Explore no more: improved high-probability regret bounds for non-stochastic bandits},
year = {2015},
publisher = {MIT Press},
address = {Cambridge, MA, USA},
booktitle = {Proceedings of the 29th International Conference on Neural Information Processing Systems - Volume 2},
pages = {3168–3176},
numpages = {9},
location = {Montreal, Canada},
series = {NIPS'15}
}

@inproceedings{Koc14,
  title     = {Efficient Learning by Implicit Exploration in Bandit Problems with Side Observations},
  author    = {Koc{\'a}k, Tom{\'a}{\v s} and Neu, Gergely and Valko, Michal and Munos, R{\'e}mi},
  booktitle = {Advances in Neural Information Processing Systems (NeurIPS) 27},
  pages     = {613--621},
  year      = {2014}
}

@book{Mie99,
  title={Nonlinear Multiobjective Optimization},
  author={Miettinen, Kaisa},
  year={1999},
  publisher={Kluwer Academic Publishers}
}

@book{Ehr05,
  title={Multicriteria Optimization},
  author={Ehrgott, Matthias},
  year={2005},
  publisher={Springer}
}

@article{Kha24,
  title={A multi-objective optimisation approach with improved pareto-optimal solutions to enhance economic and environmental dispatch in power systems},
  author={Khalil, Muhammad Ilyas Khan and Rahman, Izaz Ur and Zakarya, Muhammad and Zia, Ashraf and Khan, Ayaz Ali and Qazani, Mohammad Reza Chalak and Al-Bahri, Mahmood and Haleem, Muhammad},
  journal={Scientific Reports},
  volume={14},
  number={1},
  pages={13418},
  year={2024},
  publisher={Nature Publishing Group UK London}
}

@article{Hay21,
   title={A practical guide to multi-objective reinforcement learning and planning},
   volume={36},
   ISSN={1573-7454},
   url={http://dx.doi.org/10.1007/s10458-022-09552-y},
   DOI={10.1007/s10458-022-09552-y},
   number={1},
   journal={Autonomous Agents and Multi-Agent Systems},
   publisher={Springer Science and Business Media LLC},
   author={Hayes, Conor F. and Rădulescu, Roxana and Bargiacchi, Eugenio and Källström, Johan and Macfarlane, Matthew and Reymond, Mathieu and Verstraeten, Timothy and Zintgraf, Luisa M. and Dazeley, Richard and Heintz, Fredrik and Howley, Enda and Irissappane, Athirai A. and Mannion, Patrick and Nowé, Ann and Ramos, Gabriel and Restelli, Marcello and Vamplew, Peter and Roijers, Diederik M.},
   year={2022},
   month=Apr }

@article{RVWD13,
author = {Roijers, Diederik M. and Vamplew, Peter and Whiteson, Shimon and Dazeley, Richard},
title = {A survey of multi-objective sequential decision-making},
year = {2013},
issue_date = {October 2013},
publisher = {AI Access Foundation},
address = {El Segundo, CA, USA},
volume = {48},
number = {1},
issn = {1076-9757},
journal = {J. Artif. Int. Res.},
month = oct,
pages = {67–113},
numpages = {47}
}

@inproceedings{FKS21,
 author = {Gabriele Farina and Christian Kroer and Tuomas Sandholm},
 booktitle = {AAAI},
 date = {2021-02},
 note = {},
 title = {Faster Game Solving via Predictive Blackwell Approachability: Connecting Regret Matching and Mirror Descent},
 url = {https://arxiv.org/abs/2007.14358},
 year = {2021}
}

@article{GM25,
  title={Blackwell's Approachability with Approximation Algorithms},
  author={Garber, Dan and Massalha, Mhna},
  journal={arXiv preprint arXiv:2502.03919},
  year={2025}
}

@techreport{techreport,
    author = {Asadollahi, Ehsan and Hawkins, Calvin and Hale, Matthew},
    title = {Tech Report: Online Scalarization in Vector-Valued Games},
    institution = {Georgia Institute of Technology},
    year = {2026},
    note = {Online at: \url{https://sites.gatech.edu/ece-corelab/files/2026/03/Online_Scalarization_TechReport.pdf}}
}

@InProceedings{PM13,
  title = 	 {Approachability, fast and slow},
  author = 	 {Perchet, Vianney and Mannor, Shie},
  booktitle = 	 {Proceedings of the 26th Annual Conference on Learning Theory},
  pages = 	 {474--488},
  year = 	 {2013},
  editor = 	 {Shalev-Shwartz, Shai and Steinwart, Ingo},
  volume = 	 {30},
  series = 	 {Proceedings of Machine Learning Research},
  address = 	 {Princeton, NJ, USA},
  month = 	 {12--14 Jun},
  publisher =    {PMLR},
  url = 	 {https://proceedings.mlr.press/v30/Perchet13.html}
}

@article{Blk56,
    author = {David Blackwell},
    title = {An analog of the minimax theorem for vector payoffs},
    journal = {Pacific Journal of Mathematics},
    year = {1956},
    url = {https://msp.org/pjm/1956/6-1/pjm-v6-n1-p01-s.pdf}
}

@article{Shi16,
  author  = {Nahum Shimkin},
  title   = {An Online Convex Optimization Approach to Blackwell's Approachability},
  journal = {Journal of Machine Learning Research},
  year    = {2016},
  volume  = {17},
  number  = {129},
  pages   = {1--23},
  url     = {http://jmlr.org/papers/v17/15-339.html}
}

@InProceedings{ABH11,
  title = 	 {Blackwell Approachability and No-Regret Learning are Equivalent},
  author = 	 {Abernethy, Jacob and Bartlett, Peter L. and Hazan, Elad},
  booktitle = 	 {Proceedings of the 24th Annual Conference on Learning Theory},
  pages = 	 {27--46},
  year = 	 {2011},
  editor = 	 {Kakade, Sham M. and von Luxburg, Ulrike},
  volume = 	 {19},
  series = 	 {Proceedings of Machine Learning Research},
  address = 	 {Budapest, Hungary},
  month = 	 {09--11 Jun},
  publisher =    {PMLR},
  url = 	 {https://proceedings.mlr.press/v19/abernethy11b.html}
}

@Article{MT06,
journal={Games and Economic Behavior},
author={Mannor, Shie and Tsitsiklis, John N.},
title={Approachability in repeated games: Computational aspects and a Stackelberg variant},
year={2009},
month={May},
pages={315-325},
volume={66},
number={1},
doi={None},
url={https://ideas.repec.org/a/eee/gamebe/v66y2009i1p315-325.html},
}

@article{Per18,
title = {Approachability, regret and calibration: Implications and equivalences},
journal = {Journal of Dynamics and Games},
volume = {1},
number = {2},
pages = {181-254},
year = {2014},
issn = {2164-6066},
doi = {10.3934/jdg.2014.1.181},
url = {https://www.aimsciences.org/article/id/873959ac-8911-4a81-abb5-7562a8ffa925},
author = {Vianney Perchet}
}

@article{CKR25,
    author = {Giovanni P. Crespi, Daishi Kuroiwa & Matteo Rocca},
    title = {Vector-valued games: characterization of equilibria in matrix games},
    journal = {Mathematical Methods of Operations Research},
    year = {2025}
}

@article{shalev-shwartz_omd_hazan_bubeck,
author = {Shalev-Shwartz, Shai},
title = {Online Learning and Online Convex Optimization},
year = {2012},
issue_date = {February 2012},
publisher = {Now Publishers Inc.},
address = {Hanover, MA, USA},
volume = {4},
number = {2},
issn = {1935-8237},
url = {https://doi.org/10.1561/2200000018},
doi = {10.1561/2200000018},
journal = {Found. Trends Mach. Learn.},
month = feb,
pages = {107–194},
numpages = {88}
}
